\documentclass[final,5p,times,twocolumn]{elsarticle}

\usepackage{amsfonts,amsthm,amssymb}
\usepackage{amsmath}
\usepackage{amsfonts}
\usepackage{dsfont}
\allowdisplaybreaks
\usepackage{xcolor}
\usepackage{tikz}
\usepackage{bm}
\usepackage{multirow}
\usepackage{booktabs}
\usepackage{epstopdf}
\usepackage{array}
\usepackage{mathtools}
\mathtoolsset{showonlyrefs=true}
  \def \ind{\mathds{1}}
  \newcommand{\dd}{\mathrm{d}}

\usepackage[numbers]{natbib}
\usepackage{floatrow}
\usepackage{hyperref} 
\hypersetup{
	colorlinks   = true,
	citecolor    = blue,
	linkcolor    = red,
	urlcolor     = magenta
}

\newtheorem{theorem}{Theorem}[section]

\newtheorem{proposition}[theorem]{Proposition}

\newtheorem{lemma}[theorem]{Lemma}

\newtheorem{definition}[theorem]{Definition}
\newtheorem{corollary}[theorem]{Corollary}

\usepackage{bbm}

\journal{Operations Research Letters}

\begin{document}

\begin{frontmatter}



\title{Static Replication of Impermanent Loss for  Concentrated Liquidity Provision in   Decentralised Markets}


\author[JD]{Jun Deng}
\author[HZ]{Hua  Zong}
\author[YW]{Yun Wang}

\address[JD]{Corresponding author.   School of Finance,
	University of International Business and Economics, Beijing, China. Email: jundeng@uibe.edu.cn.}

\address[HZ]{ School of Finance,
	University of International Business and Economics, Beijing, China. Email: 202000210049@uibe.edu.cn}

\address[YW]{ School of Finance,
	University of International Business and Economics, Beijing, China. Email: ywang@uibe.edu.cn}

\begin{abstract}
This article analytically characterizes the impermanent loss of concentrated liquidity provision for   automatic market makers in decentralised markets such as Uniswap. We propose two static replication formulas  for    the  impermanent loss   by a combination of European calls or puts with strike prices supported on the liquidity provision price interval. It facilitates liquidity providers to hedge  impermanent loss by trading crypto options in more liquid centralised  exchanges such as Deribit.      Numerical examples illustrate the astonishing   accuracy of the static replication. 
\end{abstract}


\end{frontmatter}


\section{Introduction}

Decentralised exchanges (DEXs)  like Uniswap and Sushiswap facilitate  traders to swap  tokens in the listed liquidity pools by the architecture of automatic market making (AMM) without the intermediary    centralised institutions.  These exchanges  utilize    open-source protocols for providing liquidity and trading crypto tokens  and  all trades are recorded   on Ethereum blockchain.  The protocol is non-upgradable and designed to be censorship resistant without know-your-custom rule (KYC).   Instead of using limited order book as in traditional centralised financial markets that would induce extreme costly   gas fee by miners to verify transactions    on blockchain, most DEXs such as Uniswap and Sushiswap use constant product function automated market making protocol\footnote{Gas fee is  paid to miners for  validating   transactions on the Ethereum blockchain to compensate their computational resources.   
	Gas fee is often denominated in `Gwei', which is a unit of measure for the Ethereum's native currency, Ether (ETH) (1   Gwei = $10^{-9}$ ETH).}. In this paper, we would treat the dominant decentralised exchange, Uniswap that initiates  its   first version  protocol     in November 2018.  The Uniswap market lists over 400 tokens and 900 token pairs. The daily average trading volume exceeds   2 billion USD  in 2022 where the most traded pair USDC/ETH consists of  almost $50\%$  share of   total volume, followed by USDT/ETH around 10\% volume.

The constant product function automated market making protocol on Uniswap (see the v2 whitepaper \citet{adams2020uniswap}) allows traders to add, remove and swap tokens in the pool that could host any pair of tokens $(\mathtt{T_X}, \mathtt{T_Y})$.   The token $\mathtt{T_Y}$ (such as stablecoin USDC) is treated as the unit of account, i.e. the numeraire,  and    $\mathtt{T_X}$ is taken as the more volatile token  such as ETH which is the native cryptocurrency on the Ethereum.  For a pool with reserves $(X,Y)$ of tokens $(\mathtt{T_X}, \mathtt{T_Y})$, to endogenously determine the pool price,  the   pool tracks  the constant product `bonding' reserve curve $X\cdot Y = L^2$. The constant $L$, called the liquidity,  is set at the inception  and remains unchanged across trades.

To exchange $\Delta X$ amount of token $\mathtt{T_X}$ for $\Delta Y$ quantity of token $\mathtt{T_Y}$, the trader must stay on the  `bonding' curve, i.e.,    $
\left( X + \Delta X\right)\cdot \left( Y - \Delta Y\right) = L^2 = XY.
$
This  means  the trader could deposit $\Delta X$ number of  token $\mathtt{T_X}$ to swap $\Delta Y =  Y\cdot \Delta X/ (X+\Delta X)$ token $\mathtt{T_Y}$ out from the pool.   The product of new reserves  $X^{' }= X+\Delta X$ and $Y^{'}= Y -\Delta Y$ remains the constant   $X^{'}\cdot Y^{'}= L^2$.	To compensate the risk, such as impermanent loss,    taken by liquidity providers, the protocol would charge a swap fee   $\gamma \cdot \Delta X$ in terms of token sending in.  The    fee rate $\gamma$ is initiated and unchangeable when the pool is created, e.g.  $0.1\%, 0.05\%, 0.03\%$ and $0.01\%$ on Uniswap.  The constant product `bonding' curve endogenously yields a relative price in the pool.  In this article,  the \textit{pool price} of token $\mathtt{T_X}$ is  denominated in   token  $\mathtt{T_Y}$ as  $P = Y/X$.\footnote{The Uniswap protocol always tracks  liquidity $L$ and price $P$ instead of reserves $X$ and $Y$.}  In such a way, the reserves are  
$
X = L/\sqrt{P}$ and $  Y = L\sqrt{P}.$
A liquidity provider could add (remove) $\Delta L$ liquidity to the existing pool by depositing (redeeming) $\Delta X = \Delta L/ \sqrt{P}$ number of token $\mathtt{T_X}$ and $\Delta Y = \Delta L \sqrt{P}$ of token $\mathtt{T_Y}$. The liquidity provision is supplied at the current price $P$ and does not alter the pool price. 

The first and second protocol versions   of Uniswap are criticized by  low capital efficiency where liquidity provisions  are dispersed on the price range $(0,\infty)$  and only a small fraction  of total reserves is utilized during swap.  Each liquidity provider  only earns a small fraction fee proportional to her share in the pool. To promote capital efficiency through elimination of unused collateral, the  Uniswap v3 protocol was  launched on the Ethereum mainnet on May 2021 with the groundbreaking  innovative feature of  concentrated liquidity provision where liquidity providers could specify the price interval $[{\mathbb  P}_l, {\mathbb  P}_u]$ that they are willing to supply liquidity, see the whitepaper \citet{adams2021uniswap}. This resembles limit order instead of market order in previous less-efficient v2 protocol.  The `bonding' curve is shifted as $\left( X + L /\sqrt{{\mathbb  P}_l}\right)\cdot  \left( Y +  {L}{\sqrt{{\mathbb  P}_u}}\right) = L^2.$ The details are given  in the next section. 

The liquidity provider is  exposed to impermanent loss that is only realized until  depleting liquidity and withdrawing   the tokens  from the   pool.   This loss is typically calculated as the   difference of   her supplied   token pair value  in the liquidity pool  and  the value of simply holding  the tokens statically when entering the pool. Since traders always exchange less valuable token for more valuable one, liquidity providers always suffer impermanent loss (IL) that could be significant.   \citet{loesch2021impermanent} estimate     from May to Sep. 2021 the total IL  is roughly  \$260.1 million USD and 49.5\% of liquidity providers with negative returns in Uniswap v3 market.  

In this paper, we propose a static hedge strategy for liquidity providers using standard European options to eliminate the impact of   IL.  First, we   show that liquidity providers   equivalently  long and short different call and put options by liquidity provision and  explicitly characterize the  impermanent loss  as a combination of several calls and puts with different strike prices and underlying driving processes.  Second, we propose two static   replication formulas that facilitate  liquidity providers to hedge the impermanent loss risk by taking long positions of  standard European call or put options in these centralised options market such as Deribit\footnote{Deribit is the largest centralised Bitcoin and ETH options  exchanges. More information could be found at \href{www.deribit.com}{www.deribit.com}.}.  At last, we numerically verify  the static replication accuracy   that would reduce  liquidity providers' impermanent loss risk tremendously. 

We contribute to the   continually growing body of literature on decentralised exchanges in several ways. For classical market making, we refer to the seminal works of \cite{amihud1980dealership}, \cite{o1986microeconomics} and \cite{korajczyk2019high}, to name a few.  \citet{angeris2019analysis} analyze no-arbitrage boundaries and price stability in Uniswap market.    \cite{malamud2017decentralized}  show that the equilibrium utility in a decentralized market can be strictly higher   than in a centralized market and \citet{lehar2021decentralized}   propose an equilibrium model and give conditions under which the automatic market making (AMM)  dominates a limit order market and  \citet{capponi2021adoption} study the market microstructure of AMM.   Another strand is to address the optimal liquidity provision in Uniswap market, see \cite{aoyagi2020liquidity}, \cite{aigner2021uniswap} and \cite{neuder2021strategic}. 
However, they only focus on Uniswap v2 protocol either without the feature of concentrated liquidity provision or   incorporating  the impermanent loss. One exception is \citet{loesch2021impermanent} that  empirically calculate the impermanent loss (IL) in Uniswap v3 market using the on-chain data.  Our paper is more related to \citet{clark2020replicating} that studies static replication in Uniswap v2 market. To the best of our knowledge, we are the first to characterize the option-like structure of  IL that is both suffered from delta, vega and gamma exposures in Uniswap v3 market. Second, from methodological perspective, we propose  a static option replication formula for squared-root price process that is further tailored to develop our replication formulas for the impermanent loss.  It facilitates liquidity providers to hedge  permanent loss by trading crypto options in more liquid centralised  exchanges such as Deribit.

The rest of the paper is structured as follows. Section \ref{subsec_liquidity_prov} introduces the concentrated liquidity provision protocol. In Section \ref{subsec_perm_loss}, we characterize the impermanent loss and present the static replication formulas.  Section \ref{sec_empirical} demonstrates the static hedge replication accuracy and Section \ref{sec_conclusion} concludes.

\section{Concentrated Liquidity Provision  }\label{subsec_liquidity_prov}

The constant product function protocol of Uniswap   v2 facilitates  token swapers and liquidity providers to interact with the pool automatically without any financial intermediaries, although suffering   low capital efficiency.  The     Uniswap v3, launched on the Ethereum mainnet on May 2021,    has popularized the innovative feature   of concentrated  liquidity provision. This  increases the capital efficiency tremendously, up to 4000x relative to v2, at the sacrifice of  higher leverage and impermanent loss.   

When supplying  liquidity, the liquidity  provider specifies a lower price ${\mathbb  P}_l$ and a upper price ${\mathbb  P}_u$ and she  
earns    transaction fees  paid by swapers whenever the price remains in the interval   $[{\mathbb  P}_l, {\mathbb  P}_u]$. When the price moves out of the range $[{\mathbb  P}_l, {\mathbb  P}_u]$, the  position is inactive and she no longer earns any fee. Until the price re-enters   into the interval, her position is activated again.    Specifically,  the `bonding' curve  of tokens $\mathtt{T_X}$ and $\mathtt{T_Y}$ satisfy the shifted constant product function
\begin{align}
	\left( X + L /\sqrt{{\mathbb  P}_l}\right)\cdot  \left( Y +  {L}{\sqrt{{\mathbb  P}_u}}\right) = L^2.
\end{align}
The amount $L /\sqrt{{\mathbb  P}_l}$ and ${L}{\sqrt{{\mathbb  P}_u}}$ are the \textit{virtual} token reserves which are not tradable. 
Depending on the location of supported price interval $[{\mathbb  P}_l, {\mathbb  P}_u]$ relative to the current price $ {P_0}$,   to supply $\Delta L$ liquidity,  the liquidity provider's deposits $\Delta  X$ and $\Delta  Y$ of  tokens  $\mathtt{T_X}$ and $\mathtt{T_Y}$ are
\begin{align}\label{eq_dx_dy_v3}
	\Delta X &= \left\{  \begin{array}{cl}
		\Delta L \left(\frac{1}{\sqrt{{\mathbb  P}_l}}-  \frac{1}{\sqrt{{\mathbb  P}_u}}\right), & \mbox{if}\quad   P_0<{\mathbb  P}_l,\\
		\Delta L \left(\frac{1}{\sqrt{P_0}}-  \frac{1}{\sqrt{{\mathbb  P}_u}}\right), & \mbox{if}\quad   P_0 \in [{\mathbb  P}_l,  {\mathbb  P}_u],\\
		0, & \mbox{if}\quad     P_0> {\mathbb  P}_u.
	\end{array}\right. \\
	\Delta Y &= \left\{  \begin{array}{cl}
		0, & \mbox{if}\quad   P_0<{\mathbb  P}_l,\\
		\Delta L (\sqrt{P_0} - \sqrt{{\mathbb  P}_l}), & \mbox{if}\quad   P_0 \in [{\mathbb  P}_l,  {\mathbb  P}_u],\\
		\Delta L (\sqrt{{\mathbb  P}_u} - \sqrt{{\mathbb  P}_l}), & \mbox{if}\quad     P_0> {\mathbb  P}_u.
	\end{array}\right.
\end{align}
The deposits $\Delta X$ and $\Delta Y$ could be regarded as the trading volume that are needed to move price out   the supported interval $[{\mathbb  P}_l, {\mathbb  P}_u]$ from the lower boundary ${\mathbb  P}_l$ and the upper boundary  ${\mathbb  P}_u$.  The liquidity $\Delta L$ is not a tradable asset, only a synonym of token reserves $\Delta X$ and $\Delta Y$.
Several important facts follows from \eqref{eq_dx_dy_v3}.
\begin{enumerate}
	\item[(i)]   The liquidity $\Delta L$ supplied on the interval $[{\mathbb  P}_l, {\mathbb  P}_u]$ could be treated as uniformly distributed. We only illustrate one case when ${\mathbb  P}_l\leq   P_0 \leq {\mathbb  P}_u$, where $P_0$ is the current pool price. That means if we artificially split  $[{\mathbb  P}_l, {\mathbb  P}_u]$ into two sub-intervals  $[{\mathbb  P}_l,  P_0 ]$ and $[P_0, {\mathbb  P}_u]$, the liquidity  on two intervals  are both equal to $\Delta L$. By the supply equations in \eqref{eq_dx_dy_v3}, we could reformulate it as 
	\begin{align}
		\Delta X &= \underbrace{ 0}_{\Delta X^l  \ \text{deposit   on} \ [{\mathbb  P}_l,  P_0]} + \underbrace{\Delta L \left(\frac{1}{\sqrt{{P_0}}}-  \frac{1}{\sqrt{{\mathbb  P}_u}}\right)}_{\Delta X^r \ \text{deposit   on} \  [P_0, {\mathbb  P}_u]},\\
		\Delta Y &= \underbrace{\Delta L (\sqrt{P_0} - \sqrt{{\mathbb  P}_l})}_{\Delta Y^l \ \text{deposit  on} \ [{\mathbb  P}_l,  P_0]}
		+ \underbrace{0}_{\Delta Y^r\ \text{deposit   on }\  [P_0, {\mathbb  P}_u]}.
	\end{align}
	It clearly shows that the   provider supplies liquidity $\Delta L$ on the left interval $[{\mathbb  P}_l,  P_0]$ with reserves $(\Delta X^l, \Delta Y^l)$ and on the right interval $[P_0, {\mathbb  P}_u]$  with reserves $(\Delta X^r, \Delta Y^r)$ simultaneously. 
	\item[(ii)] The discussion above motivates us   considering liquidity provision only from two sides of the current price $P_0$ that means we treat the liquidity provision on ${\mathbb  P}_l\leq  {P_0} \leq {\mathbb  P}_u$ as two independent and disjoint price bins $[{\mathbb  P}_l,  {P_0}]$ and $[{P_0}, {\mathbb  P}_u]$. In doing so, on the lower price bin $[{\mathbb  P}_l, {P_0}]$  she only supplies  token $\mathtt{T_Y}$ with the amount of $\Delta L (\sqrt{{P_{0}}} - \sqrt{{\mathbb  P}_l})$, in the meanwhile, she only deposits  token $\mathtt{T_X}$ with the amount of  $\Delta L \left(\frac{1}{\sqrt{{P_0}}}-  \frac{1}{\sqrt{{\mathbb  P}_u}}\right)$ in the upper price bin. The two price intervals $[{\mathbb  P}_l,  {P_0}]$ and $[{P_0},  {\mathbb  P}_u]$ resembles the bid-ask  prices of the traditional limit-order-book. This would greatly simply the analysis of   impermanent loss below. 
\end{enumerate}

\section{Impermanent Loss of Liquidity Provision}\label{subsec_perm_loss}
The  impermanent loss     is the possible loss from liquidity provision, compared to the static strategy where the liquidity provider holds the corresponding tokens in the pool unchanged.  Due to the change of token pool price, once the provider closes her position and exits the pool, the  impermanent loss would be realized.  Following \cite{aigner2021uniswap},  \cite{heimbach2022risks}, and  \cite{loesch2021impermanent},  we define the     impermanent loss (IL)  as follows.
\begin{definition}\label{definition_imperm_loss}
	For a  liquidity provision  with  deposits   $X_{0}$ and $Y_{0}$  of tokens  $\mathtt{T_X}$ and   $\mathtt{T_Y}$ at initial time $0$, 
	the realized impermanent loss (IL)  at time $t$ when   removing  the liquidity  is the capital loss   if she holds token pair statically   at initial time $0$ instead.  Specifically, the   impermanent loss (IL) is calculated as 
	\begin{align}
		\text{IL} = Y_{t}  - Y_{0}  + (X_{t}  - X_{0}) \cdot P_{t}.
	\end{align}
	Here,  $X_t$ and $Y_t$ are the quantities withdrawn at time $t$  and $P_t$ is the price of token $\mathtt{T_X}$.
\end{definition}
The  impermanent loss is not defined as the difference between  `exit' value and `entry' value as usual. Here, we take the same perspective of \cite{aigner2021uniswap},  \cite{heimbach2022risks}, and  \cite{loesch2021impermanent} and industry practice that liquidity providers treat  impermanent loss as the cost of buying their initial liquidity deposits back when exiting the pool. { A recent work on other form of arbitrage and losses in decentralised exchanges (the so-called miner extractable value) is studied in \citet{capponi2022evolution}.}  Since    token $\mathtt{T_Y}$ is usual some stable-coin (such as USDC and USDT) with value pegged to 1 USD, we always consider the realized impermanent loss (IL) in terms of token $\mathtt{T_Y}$.

From the analysis in Section \ref{subsec_liquidity_prov}, we could treat each liquidity provision  separately from two sides of the current pool price $ {P_0}$.  Without loss of generality,  we assume the liquidity $\Delta L$ is supplied on the right side price  interval   $[{\mathbb  P}_l, {\mathbb  P}_u]$ where ${P_0}\leq {\mathbb  P}_l\leq {\mathbb  P}_u$ that resembles the ask prices in limit order book. From \eqref{eq_dx_dy_v3}, the number of tokens   required to establish the position is 
\begin{align}
	\Delta Y_0 = 0, \qquad   \Delta X_0 = \Delta L \left(\frac{1}{\sqrt{{\mathbb  P}_l}}-  \frac{1}{\sqrt{{\mathbb  P}_u}}\right).
\end{align}
Now, we track the token holdings when the provider closes her position where the price changes from $ {P_0}$ to $P_t$ at exiting time $t$.  We distinguish three possible locations of price $P_{t}$. 
\begin{enumerate}
	\item[(i)] $P_{t}\in [{\mathbb  P}_l, {\mathbb  P}_u]:$ When the price increases and moves into the price interval $[{\mathbb  P}_l, {\mathbb  P}_u]$, the initial reserve $\Delta X_0$ is partially converted to \textit{less valuable}  token $\mathtt{T_Y}$ that would induce   loss to the liquidity provider.    At time $t$, when exiting, the quantities  of tokens $\mathtt{T_Y}$ and $\mathtt{T_X}$ can be retrieved from the pool are    
	\begin{align}
		\Delta Y_{t}= \Delta L(\sqrt{P_{t}} - \sqrt{{\mathbb  P}_l}), \quad 
		\Delta X_{t} = \Delta L\left( \frac{1}{\sqrt{P_{t}}} - \frac{1}{\sqrt{{\mathbb  P}_u}}\right).
	\end{align}
	The amount of $\Delta X_0 - \Delta X_t$ of token $\mathtt{T_X}$ is converted to token $\mathtt{T_Y}$ that would have been sell more   if she does not provide liquidity. 
	Therefore, the IL denominated in token  $\mathtt{T_Y}$ is  
	\begin{align}
		\text{IL}^{(1)} &=  \Delta L(\sqrt{P_{t}} - \sqrt{{\mathbb  P}_l})
		+  \Delta L\left( \frac{1}{\sqrt{P_{t}}} - \frac{1}{\sqrt{{\mathbb  P}_l}}\right) P_{t}  \\
		&=  \Delta L\left(2  \sqrt{P_{t}} 
		-  \frac{P_{t}}{\sqrt{{\mathbb  P}_l}}
		-  \sqrt{{\mathbb  P}_l}    
		\right)\leq   0.
	\end{align}
	
	With the escalating of token $\mathtt{T_X}$'s price, the provider is consistently and gradually selling more  valuable token $\mathtt{T_X}$ for token $\mathtt{T_Y}$. Therefore, she suffers a loss.
	\item[(ii)] $P_{t} \geq {\mathbb  P}_u:$ When the price crosses above the upper price ${\mathbb  P}_u$, all token reserve $\Delta X_0$ are converted to token $\mathtt{T_Y}$.  The amount of tokens $\mathtt{T_Y}$ and $\mathtt{T_X}$ can be retrieved from the pool at time $t$ are    
	\begin{align}
		\Delta Y_{t} = \Delta L(\sqrt{{\mathbb  P}_u} - \sqrt{{\mathbb  P}_l}), \quad 
		\Delta X_{t} = 0.
	\end{align}
	Therefore, the  IL  is
	\begin{align}
		\text{IL}^{(2)} &=  \Delta L\left(  \sqrt{{\mathbb  P}_u} - \sqrt{{\mathbb  P}_l} -
		\left(\frac{1}{\sqrt{{\mathbb  P}_l}} - \frac{1}{\sqrt{{\mathbb  P}_u}}\right)  {P_t} \right)  
		\leq  0
	\end{align}
	In the meanwhile, the average selling price of token $\mathtt{T_X}$ is $\Delta Y_t/\Delta X_0 = \sqrt{{\mathbb  P}_u {\mathbb  P}_l} $, lower than the market price $P_t$.
	\item[(iii)] $P_{t} \leq  {\mathbb  P}_l:$ When the price stays below the lower boundary ${\mathbb  P}_l$,   the quantity of token $\mathtt{T_X}$ is unchanged and  the IL  is  0.
\end{enumerate}

Taking together, the aggregated impermanent loss from the right side price interval $\text{IL}^{\mathtt{R}}$   is 
\begin{align}\label{eq_losP_for_rightside}
	\frac{\text{IL}^{\mathtt{R}}}{\Delta L} &=  
	\left(2  \sqrt{P_{t}} 
	-  \frac{P_{t}}{\sqrt{{\mathbb  P}_l}}
	-  \sqrt{{\mathbb  P}_l}    
	\right) \ind_{\{{\mathbb  P}_l\leq P_{t} \leq {\mathbb  P}_u\}}  \\
	&+  \left(  \sqrt{{\mathbb  P}_u} - \sqrt{{\mathbb  P}_l} -
	\left(\frac{1}{\sqrt{{\mathbb  P}_l}} - \frac{1}{\sqrt{{\mathbb  P}_u}}\right)  {P_t} \right)  
	\ind_{  \{P_{t} \geq {\mathbb  P}_u\}}. 
\end{align}
Similar argument to the liquidity provision from the left side price interval  $[\mathbb{S}_l, \mathbb{S}_u]$ (i.e. $\mathbb{S}_l\leq \mathbb{S}_u\leq P_0$) leads to 
\begin{align}\label{eq_losP_for_leftside}
	\frac{\text{IL}^{\mathtt{L}}}{\Delta L} &=  
	\left(2  \sqrt{P_{t}} 
	-  \frac{P_{t}}{\sqrt{{\mathbb  S}_u}}
	-  \sqrt{{\mathbb  S}_u}    
	\right) \ind_{\{{\mathbb  S}_l\leq P_{t} \leq {\mathbb  S}_u\}}   \\
	& +  \left(  \left(\frac{1}{\sqrt{{\mathbb  S}_l}} - \frac{1}{\sqrt{{\mathbb  S}_u}}\right)  {P_t} - \sqrt{{\mathbb  S}_u}+ \sqrt{{\mathbb  S}_l}  
	\right)  
	\ind_{  \{P_{t} \leq  {\mathbb  S}_l\}}. 
\end{align}
We call the ratio $\text{IL} /\Delta L$ as \textit{unit impermanent loss per liquidity} (UIL). Rearranging  \eqref{eq_losP_for_rightside} and \eqref{eq_losP_for_leftside} gives the following representation. \footnote{  If we track the impermanent loss by trading volume and there are multiple liquidity providers with possible overlap provision price intervals, we could split them to non-overlap ones and investigate impermanent loss separately. We thank  the referee for pointing out this.  In this paper, we take the perspective of tracing price changes that simplifies the exposure. }

\begin{proposition}\label{proposition_UIL_representation}
	The impermanent loss per liquidity (UIL) is characterized as a combination of short and long positions in different call  options as following. 
	\begin{align}\label{eq_UIL_representation}
		\rm{UIL}^{\mathtt{R}} &= 		2 \left(\sqrt{P_t} - \sqrt{{\mathbb  P}_l}\right)^{+} 
		- 2 \left(\sqrt{P_t} - \sqrt{{\mathbb  P}_u}\right)^{+}
		\\
		&- \frac{1}{\sqrt{{\mathbb  P}_l}}  	 \left( {P_t} -  {{\mathbb  P}_l}\right)^{+}
		+ \frac{1}{\sqrt{{\mathbb  P}_u}}  	 \left( {P_t} -  {{\mathbb  P}_u}\right)^{+},\\
		\rm{UIL}^{\mathtt{L}} &= 		2 \left(\sqrt{{\mathbb  S}_l} - \sqrt{P_t} \right)^{+} 
		- 2 \left(\sqrt{{\mathbb  S}_u} - \sqrt{P_t}\right)^{+}
		\\
		&- \frac{1}{\sqrt{{\mathbb  S}_l}}  	 \left( {{\mathbb  S}_l} - {P_t}   \right)^{+}
		+ \frac{1}{\sqrt{{\mathbb  S}_u}}  	 \left(   {{\mathbb  S}_u} -  {P_t}\right)^{+}. \label{eq_UIL_representation_left}
	\end{align}
\end{proposition}
Proposition \ref{proposition_UIL_representation} demonstrates the unit  impermanent loss $\text{UIL}^{\mathtt{R}}$  ($\text{UIL}^{\mathtt{L}}$) is equivalent to hold  two types of  long call (put) options with price  process $P_t$ and squared root price process $\sqrt{P_t}$ and also two short call (put) positions.  The strike prices are either the lower boundary $\mathbb{P}_l$ or the upper boundary $\mathbb{P}_u$. This also means liquidity providers suffer all standard European option's risk such as vega and  gamma exposures. 

Assuming risk-free rate is zero and $P_t$ obeys geometric Brownian motion, we have the following corollary. 
\begin{corollary}
	If the price $P_t$ of token $\mathtt{T_X}$ is driven by a geometric Brownian motion with   volatility $\sigma$,  the   impermanent loss per liquidity $\text{UIL}^{\mathtt{R}}$ and $\text{UIL}^{\mathtt{L}}$ are
	\begin{align}
		{\mathbb{E}[\rm{UIL}^{\mathtt{R}}]}  &= 2\sqrt{P_0}\exp\left({-\frac{\sigma^2 t}{8}}\right) 
		\left(N(d_l +\sigma\sqrt{t}/2)-N(d_u+\sigma\sqrt{t}/2)\right) \\
		& - \sqrt{\mathbb{P}_l} N(d_l) + \sqrt{\mathbb{P}_u} N(d_u)	 \\
		&-\frac{1}{\sqrt{{\mathbb  P}_l}} P_0 N(d_l+\sigma \sqrt{t}) 
		+  \frac{1}{\sqrt{{\mathbb  P}_u}}P_0 N(d_u+\sigma \sqrt{t}), \\
		%
		{\mathbb{E}[\rm{UIL}^{\mathtt{L}}]}  &=  2\sqrt{P_0}\exp\left({-\frac{\sigma^2 t}{8}}\right) 
		\left(-N(-q_l -\sigma\sqrt{t}/2) \right.\\
		&\left. + N(-q_u-\sigma\sqrt{t}/2)\right) 
		+ \sqrt{\mathbb{S}_l} N(-q_l) - \sqrt{\mathbb{S}_u} N(-q_u)	 \\
		&+\frac{1}{\sqrt{{\mathbb  S}_l}} P_0 N(-q_l-\sigma \sqrt{t}) 
		- \frac{1}{\sqrt{{\mathbb  S}_u}}P_0 N(-q_u-\sigma \sqrt{t}).
	\end{align}
	Here, $N(\cdot)$ is the standard normal cumulative distribution function and 
	\begin{align}
		d_u &= \frac{\ln(P_0/{\mathbb  P}_u )  -  \sigma^2 t/2  }{\sigma\sqrt{t}}, \quad 
		d_l = \frac{\ln( {P_0}/{\mathbb  P}_l)  - \sigma^2 t/2}{\sigma\sqrt{t}},\\ 
		q_u &= \frac{\ln(P_0/{\mathbb  S}_u )  -  \sigma^2 t/2  }{\sigma\sqrt{t}}, \quad 
		q_l = \frac{\ln( {P_0}/{\mathbb  S}_l)  - \sigma^2 t/2}{\sigma\sqrt{t}}.
	\end{align}
\end{corollary}

\begin{figure}[]
	\centering
	\caption{Impermanent Loss Sensitives}
	\includegraphics[trim = 3cm 1cm 3cm 1cm, scale=0.45]{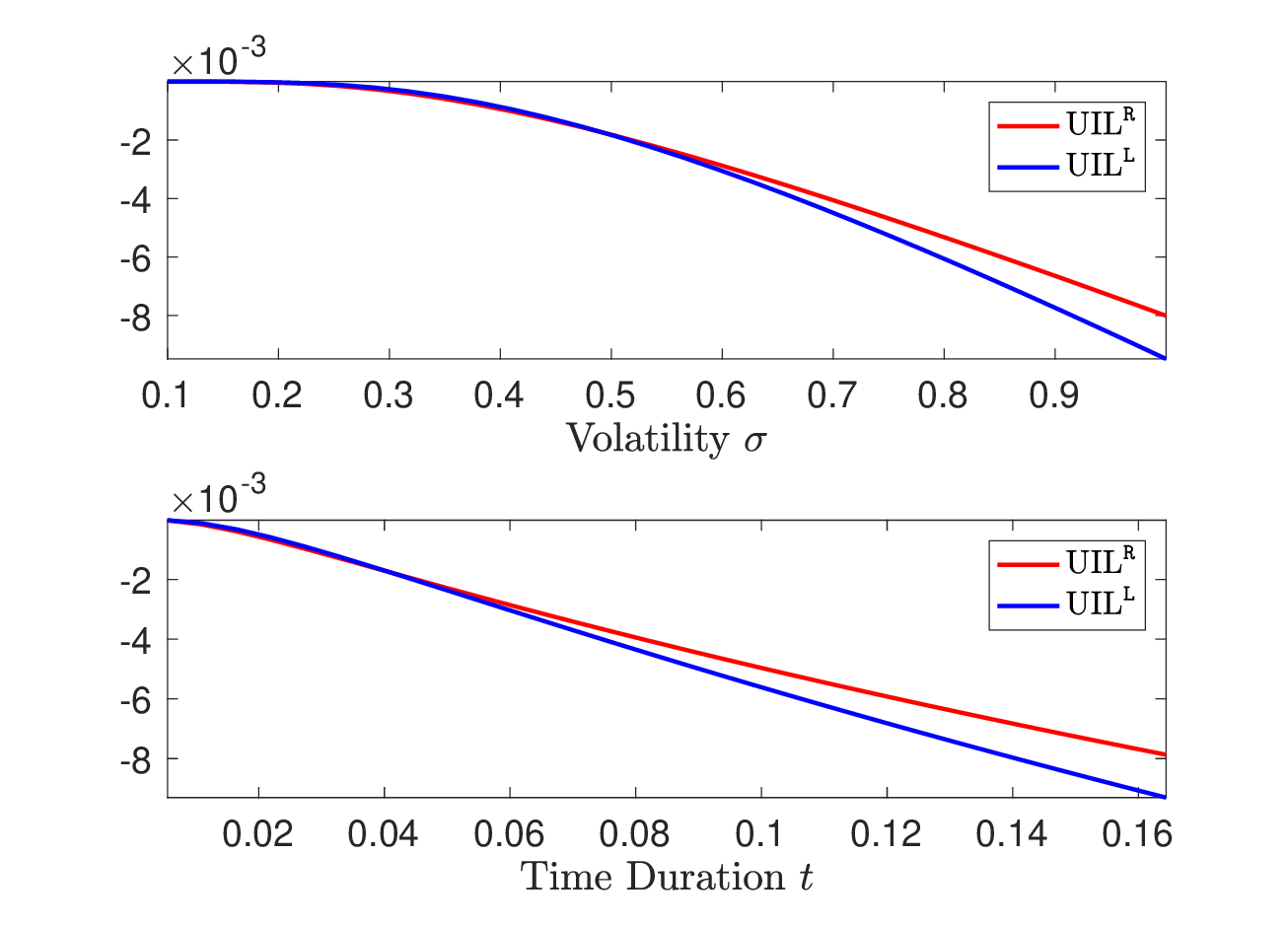}
	\label{fig_permanent_loss}
		\floatfoot{Note. Here, we set $\Delta L = 1, \ \sigma = 0.7$. The initial pool price is  $P_0=10$ and the liquidity provider supplies on the upper price interval $[{\mathbb  P}_l, {\mathbb  P}_u] = [11,12]$ and lower price interval $[{\mathbb  S}_l, {\mathbb  S}_u] = [8,9]$ and closes her position after one month, i.e. $t=30$ days.}
\end{figure}
Figure \ref{fig_permanent_loss} shows the impermanent  loss  ${\mathbb{E}[\text{UIL}^{\mathtt{R}}]} $  and ${\mathbb{E}[\text{UIL}^{\mathtt{L}}]} $ are both declining  with the increasing of volatility and exiting time $t$.  It all shows the asymmetric pattern that the right side liquidity provision  ${\mathbb{E}[\text{UIL}^{\mathtt{R}}]} $ is less sensitive to vega and theta risk.

\subsection{Static Replication of Impermanent Loss}
{  Proposition  \ref{proposition_UIL_representation} demonstrates the impermanent loss is an ``option-like'' instrument that can not be easily hedged by underlying asset and futures.  In this section, we statically replicate (or hedge) the impermanent loss by standard European call or put options. Here, static replication means the liquidity provider could buy a combination of calls or puts at inception and hold the position statically until removing her liquidity from the pool. This lowers transaction and re-balancing cost for them.} We start with two basic   equalities. 

\begin{lemma}\label{lemma_root_option_represetntation}
	The following two equations   hold.
	\begin{align}\label{eq_root_option_represetntation}
		\left(\sqrt{x} - \sqrt{\widehat{K}} \right)^+  
		&= \frac{\widehat{K}^{-\frac{1}{2}}}{2 } \left( x - \widehat{K}\right)^{+}
		-\frac{1}{4} \int_{\widehat{K}}^{\infty} K^{-\frac{3}{2}} \left(x - K\right)^+ \dd K,  \ \\
		\left(\sqrt{\widehat{K}} - \sqrt{x}   \right)^+  
		&= \frac{\widehat{K}^{-\frac{1}{2}} }{2 }\left( \widehat{K} - x  \right)^{+}
		+ \frac{1}{4} \int_{0}^{\widehat{K}} K^{-\frac{3}{2}} \left(K - x \right)^+ \dd K.
		\label{eq_root_option_represetntation_put}
	\end{align}
\end{lemma}
\begin{proof}
	Any twice differentiable function $f$ can be represented as (see \citet{carr2001towards})
	\begin{align}\label{eq_well_f_rep}
		f(x) &= f(x^\star) +f^{'}(x^\star) (x - x^\star) \\
		&+\int_0^{x^\star} f^{''}(K)(K-x)^{+} \dd K + \int_{x^\star}^{+\infty} f^{''}(K)(x - K)^{+} \dd K.
	\end{align}
	Applying formula \eqref{eq_well_f_rep} to $f(x) = \sqrt{x} - \sqrt{\widehat{K}}$ and choose $x^\star = \widehat{K}$, we have 
	\begin{align}
		\sqrt{x} - \sqrt{\widehat{K}} 
		&= \frac{1}{2 }\widehat{K}^{-\frac{1}{2}} \left( x - \widehat{K}\right) 
		- \frac{1}{4} \int_{0}^{\widehat{K}} K^{-\frac{3}{2}} \left(K - x\right)^+ \dd K \\
		&-\frac{1}{4} \int_{\widehat{K}}^{+\infty} K^{-\frac{3}{2}} \left(x - K\right)^+ \dd K.
	\end{align}
	Multiplying  $\ind_{\{x\geq \widehat{K}\}}$ in the above equation, it yields \eqref{eq_root_option_represetntation}. Similar argument leads to \eqref{eq_root_option_represetntation_put}. This completes the proof. 
\end{proof}

\begin{proposition}\label{prop_mean_UIL_representation}
	The impermanent loss per liquidity (UIL) can be statically  replicated by 
	\begin{align}\label{eq_mean_UIL_representation}
		\mathbb{E}[{\rm{UIL}}^{\mathtt{R}}] & =  
		- \frac{1}{2} \int_{{\mathbb  P}_l}^{{\mathbb  P}_u} K^{-\frac{3}{2}}  {\mathbf{C}}(K) \dd K,  \\ 
		\mathbb{E}[{\rm{UIL}}^{\mathtt{L}}]  &=  
		- \frac{1}{2} \int_{{\mathbb  S}_l}^{{\mathbb  S}_u} K^{-\frac{3}{2}}  \mathbf{P}(K) \dd K.\label{eq_mean_UIL_representation_left}
	\end{align}
	Here, $\mathbf{C}(K)$ and $\mathbf{P}(K)$ are European call and put option prices with maturity $t$ and strike price $K$. \\
	
	Especially, when the provider supplies liquidity over $(0,\infty)$ as in Uniswap v2, the total UIL is 
	$ 	- \frac{1}{2} \int_{P_0}^{\infty} K^{-\frac{3}{2}}  {\mathbf{C}}(K) \dd K
		- \frac{1}{2} \int_{0}^{P_0} K^{-\frac{3}{2}}  \mathbf{P}(K) \dd K$.
\end{proposition}

\begin{proof}
	Using equations \eqref{eq_root_option_represetntation} and \eqref{eq_UIL_representation} of Lemma \ref{lemma_root_option_represetntation} and  Proposition \ref{proposition_UIL_representation}    
	with $x = P_t$ and $\widehat{K} \in  \{{\mathbb  P}_l, {\mathbb  P}_u\}$, we have 
	\begin{align} 
		\text{UIL}^{\mathtt{R}} &= 		2 \left(\sqrt{P_t} - \sqrt{{\mathbb  P}_l}\right)^{+} 
		- 2 \left(\sqrt{P_t} - \sqrt{{\mathbb  P}_u}\right)^{+}
		\\
		&- \frac{1}{\sqrt{{\mathbb  P}_l}}  	 \left( {P_t} -  {{\mathbb  P}_l}\right)^{+}
		+ \frac{1}{\sqrt{{\mathbb  P}_u}}  	 \left( {P_t} -  {{\mathbb  P}_u}\right)^{+}\\
		& =  {\mathbb  P}_l^{-\frac{1}{2}} \left( P_t -{\mathbb  P}_l\right)^{+}
		-\frac{1}{2} \int_{{\mathbb  P}_l}^{+\infty} K^{-\frac{3}{2}} \left(P_t- K\right)^+ \dd K 
		\\
		&- {\mathbb  P}_u^{-\frac{1}{2}} \left( P_t -{\mathbb  P}_u \right)^{+}
		+\frac{1}{2} \int_{{\mathbb  P}_u}^{+\infty} K^{-\frac{3}{2}} \left(P_t - K\right)^+ \dd K\\
		&- {\mathbb  P}_l^{-\frac{1}{2}} \left( P_t -{\mathbb  P}_l\right)^{+}
		+ {\mathbb  P}_u^{-\frac{1}{2}} \left( P_t -{\mathbb  P}_u \right)^{+} \\
		& = -\frac{1}{2} \int_{{\mathbb  P}_l}^{{\mathbb  P}_u} K^{-\frac{3}{2}} \left(P_t- K\right)^+ \dd K.
	\end{align}
	Taking expectation under risk-neutral probability gives the first equation in   \eqref{eq_mean_UIL_representation}. Similar argument with  \eqref{eq_UIL_representation_left} and  \eqref{eq_root_option_represetntation_put}    gives the second equality in \eqref{eq_mean_UIL_representation_left}. When the provider supplies liquidity over $(0,\infty)$ as in Uniswap v2, the supported price intervals are  $[{\mathbb  P}_l, {\mathbb  P}_u] = [P_0, +\infty]$ and $[{\mathbb  S}_l, {\mathbb  S}_u] = [0, P_0]$. This completes the proof. 
\end{proof}

First of all,   Proposition \ref{prop_mean_UIL_representation} shows  the impermanent loss could be perfectly replicated via a group of call or put options with strike prices supported in the liquidity provision interval. This equips liquidity providers the vehicle to hedge the IL by trading options in the more liquid centralised cyptocurrency  market such as Deribit.  Second, the impermanent losses  inherit all option Greeks such as delta, gamma and vega risk factors. For instance, the delta $\partial \mathbb{E}[{\rm{UIL}}^{\mathtt{R}}] /\partial P  =  
- \frac{1}{2} \int_{{\mathbb  P}_l}^{{\mathbb  P}_u} K^{-\frac{3}{2}}  \Delta_{\mathbf{C}}(K) \dd K$ and  $\Delta_{\mathbf{C}}(K)$ is the call option's delta at strike price $K$. {  Third, in practice,  only limited option strikes are traded in the market.} Therefore, we use a discrete   version of  formulae \eqref{eq_mean_UIL_representation} and \eqref{eq_mean_UIL_representation_left}, i.e. 
$
\mathbb{E}[{\rm{UIL}}^{\mathtt{R}}] \simeq -0.5   \sum_i K_i^{-3/2} \mathbf{C}(K_i) \Delta K_i
$
and 
$
\mathbb{E}[{\rm{UIL}}^{\mathtt{L}}] \simeq -0.5   \sum_i M_i^{-3/2} \mathbf{P}(M_i) \Delta M_i.
$
Here, $(K_i)$ and $(M_i)$ are the partitions of intervals $[{\mathbb  P}_u, \ {\mathbb  P}_l]$ and $[{\mathbb  S}_u, \ {\mathbb  S}_l]$ accordingly. {The approximation would induce discretization errors for the replications which is evaluated in the following section.}

\section{Numerical Analysis}\label{sec_empirical}
\subsection{ Discretization Errors}
In this section, we assume the pool price   is driven by the Heston process 
\begin{align}
	\dd P_t & = \mu P_t \dd t + \sqrt{\nu_t} P_t \dd W^P_t, \quad 
	\dd \nu_t  = \kappa (\theta - \nu_t) \dd t  + \xi \sqrt{\nu_t} \dd W^\nu_t.
\end{align}
Here, $W^P$ and $W^\nu$ are two correlated Brownian motions with correlation $\rho$.  In the numerical analysis, we assume the initial pool price $P_0 = 10$ and volatility $\nu_0 = 0.3$ and set $\mu = 0.1 $ and $\rho = -0.3$. 
The mean-reversion speed, volatility level and volatility-of-volatility are set as  $\kappa = 0.4, \theta = 0.4$ and $\xi = 0.15$ accordingly. 
The liquidity provider supplies liquidity in  the upper price interval  $[{\mathbb  P}_l, {\mathbb  P}_u] = [11,14]$ and lower price interval $[{\mathbb  S}_l, {\mathbb  S}_u] = [6,9]$ and closes her position after one week, i.e. $t=7$ days.

First, we use Monte Carlo method to estimate  $\mathbb{E}[{\rm{UIL}}^{\mathtt{R}}], \ \mathbb{E}[{\rm{UIL}}^{\mathtt{L}}]$ and call and put   prices. {Second,  the simple trapezoidal numerical integration is utilized to approximate the integrals in \eqref{eq_mean_UIL_representation}  and \eqref{eq_mean_UIL_representation_left} with  100 different strikes   in the   intervals $[{\mathbb  P}_l, {\mathbb  P}_u]$ and  $[{\mathbb  S}_l, {\mathbb  S}_u]$}.       Table \ref{tab_accuracy} reports the replication error ratios. It shows in all scenarios the replication formulas  yield highly accurate approximations for both right and left side impermanent losses if there are enough traded option strikes. The error ratios are   roughly 0.1 base point (0.01\%) for right side impermanent losses and 0.01bp for the left one.    With the increasing of volatility level $\theta$, it also increases   impermanent losses $ \mathbb{E}[{\rm{UIL}}^{\mathtt{R}}]$ and  $\mathbb{E}[{\rm{UIL}}^{\mathtt{L}}]$  simultaneously. The effects of reversion speed $\kappa$ and volatility-of-volatility $\xi$ are mixed. 

\begin{table}[h!]
	\caption{Static Replication Accuracy for Impermanent Loss per Liquidity}
	\label{tab_accuracy}
\begin{tabular}{cccc}
	\toprule
	& $\mathbb{E} [\text{UIL}^{\mathtt{R}}]$ & $\text{UIL}^{\mathtt{R}}$ Replication & Error Ratio \\
	\toprule
	$\kappa= 0.3$ & -0.424263                              & -0.424267                             & 1.03E-05    \\
	$\kappa= 0.4$ & -0.420756                              & -0.420761                             & 1.03E-05    \\
	$\kappa= 0.5$ & -0.439243                              & -0.439248                             & 1.02E-05    \\ \\
	$\theta= 0.3$ & -0.411396                              & -0.411401                             & 1.08E-05    \\
	$\theta= 0.4$ & -0.460699                              & -0.460703                             & 1.01E-05    \\
	$\theta= 0.5$ & -0.501689                              & -0.501694                             & 9.68E-06    \\ \\
	$\xi= 0.1$    & -0.440122                              & -0.440126                             & 1.02E-05    \\
	$\xi= 0.15$   & -0.497540                              & -0.497545                             & 9.97E-06    \\
	$\xi= 0.2$    & -0.467309                              & -0.467313                             & 1.02E-05    \\
	\toprule
	& $\mathbb{E} [\text{UIL}^{\mathtt{L}}]$ & $\text{UIL}^{\mathtt{L}}$ Replication & Error Ratio \\
	$\kappa= 0.3$ & -0.177913                              & -0.177913                             & 1.59E-06    \\
	$\kappa= 0.4$ & -0.186653                              & -0.186652                             & 1.83E-06    \\
	$\kappa= 0.5$ & -0.192570                              & -0.192570                             & 1.41E-06    \\ \\
	$\theta= 0.3$ & -0.150062                              & -0.150061                             & 1.91E-06    \\
	$\theta= 0.4$ & -0.185478                              & -0.185477                             & 1.58E-06    \\
	$\theta= 0.5$ & -0.217112                              & -0.217112                             & 7.72E-07    \\ \\
	$\xi= 0.1$    & -0.187169                              & -0.187168                             & 1.72E-06    \\
	$\xi= 0.15$   & -0.182929                              & -0.182929                             & 1.36E-06    \\
	$\xi= 0.2$    & -0.184267                              & -0.184267                             & 1.18E-06   \\
		\toprule
	\end{tabular}
	\floatfoot{Note. The initial pool price and volatility are  $P_0 = 10$ and  $\nu_0 = 0.3$. Choose   $\mu = 0.1 $ and $\rho = -0.3$.   The mean-reversion speed, volatility level and volatility-of-volatility are set as  $\kappa = 0.4, \theta = 0.4$ and $\xi = 0.15$. The liquidity provider supplies liquidity in  the upper price interval  $[{\mathbb  P}_l, {\mathbb  P}_u] = [11,14]$ and lower price interval $[{\mathbb  S}_l, {\mathbb  S}_u] = [6,9]$ and closes her position after one month, i.e. $t=7$ days. }
\end{table}

\subsection{Replication Error on Deribit  Option Market}
In this section, we verify the replication accuracy using    bitcoin options traded on Deribit exchange from 1, Jan. 2020 to 31, Dec. 2020.
We access the tick-by-tick options data through Deribit API (application programming interface) which consists of 1,316,050 trades with the annual total volume of  56.7 billion USD. 

On each day, we suppose the liquidity provider could randomly enter the Uniswap market and provide liquidity for the BTC-USDC pool from the right and left sides of current price, i.e., $ [P_0, u\cdot P_0] $ and $ [d\cdot P_0, P_0] $ and  deplete liquidity until time $ T $ (1 or 2 weeks).
Here, $ P_0 $ is the entry bitcoin price and two constants $ u>1 $ and $ d<1 $ control the liquidity provision price intervals.  The impermanent loss when exiting  is calculated through equations 
\eqref{eq_UIL_representation} and \eqref{eq_UIL_representation_left}.

In the meanwhile, the liquidity provider also longs a combination of calls (or puts) with strikes in $ [P_0, u\cdot P_0] $ (or $ [d\cdot P_0, P_0] $) and holds statically to maturity $ T $. The call and put options gain are calculated via the means of 
\begin{align}
	   &-0.5   \sum_i K_i^{-3/2}(P_T - K_i)^+ \Delta K_i, \quad \mbox{and}, \\
	   &
	  -0.5   \sum_i M_i^{-3/2}  (M_i - P_T)^+ \Delta M_i.
\end{align}
Here, $ P_T $ is the bitcoin price at time $T$ and $ K_i, M_i $ are all traded option strikes.

Table \ref{tab_deribit_replication} shows several interesting facts: (1) the impermanent losses $\rm{UIL}^{\mathtt{R}}$ and $\rm{UIL}^{\mathtt{L}}$ both increase with more wider liquidity provision intervals (larger $ u $ and smaller $ d $);
(2) the right side loss  $\rm{UIL}^{\mathtt{R}}$ is more severe than left side loss $\rm{UIL}^{\mathtt{L}}$. The intuition is  the bitcoin price soars from 4,000 to nearly 30,000 in   2020 and right side liquidity provision would be much more risky; (3) even-though only a few   strikes (3, 7, 10) are traded, the replication errors are reasonable, especially for this particular volatile markets; (4) we do not observe obvious patterns when the liquidity provision duration is longer.

\begin{table}[h!]
	\caption{Static Replication using Deribit Option  }
	\label{tab_deribit_replication}
	\begin{tabular}{ccccc}
		\toprule
		T(week)            & Price Interval & $\mathbb{E} [\text{UIL}^{\mathtt{R}}]$ & Static & \#Strikes \\
		\toprule
		\multirow{3}{*}{1} & $u=1.1$ &  -1.588    &   -0.991 & 3 \\
		  & $u=1.2$        & -2.735                  & -2.294 & 7        \\
		& $u=1.3$        & -3.597                  & -3.253 & 10        \\
		\\
		\multirow{3}{*}{2} & $u=1.1$ &  -1.586    &   -0.921 & 3 \\
	   & $u=1.2$        & -2.726                  & -2.216 & 7       \\
		& $u=1.3$        & -3.569                  & -3.157 & 10       \\
		\toprule
		&                & $\mathbb{E} [\text{UIL}^{\mathtt{L}}]$ & Static & \#Strikes \\
		\multirow{3}{*}{1} & $d=0.9$ & -0.097 &  -0.066 & 3 \\
		& $d=0.8$        & -0.154                  & -0.122 & 7        \\
		& $d=0.7$        & -0.182                  & -0.149 & 10        \\
		\\
		\multirow{3}{*}{2} & $d=0.9$ & -0.101 &  -0.068 & 3 \\
		  & $d=0.8$        & -0.158                  & -0.124 & 7       \\
		& $d=0.7$        & -0.187                  & -0.154 & 10 \\
		\toprule
	\end{tabular}
	\floatfoot{Note.  The liquidity provider provides liquidity for the BTC-USDC pool from the right and left sides of current price, i.e., $ [P_0, u\cdot P_0] $ and $ [d\cdot P_0, P_0] $ and  deplete liquidity until time $ T $ (1 or 2 weeks).
		The two constants $ u>1 $ and $ d<1 $ control the liquidity provision price intervals. In the meanwhile, the liquidity provider also trades Deribit bitcoin options via the static replication formulae \eqref{eq_mean_UIL_representation} and \eqref{eq_mean_UIL_representation_left}, reported in column ``static''. The column ``\#Strikes'' is the average number of traded option strikes in the provision price intervals. }
\end{table}

\section{Conclusion}\label{sec_conclusion}
Liquidity providers supply crypto tokens  from the right and left sides of the current price in decentralised markets such as Uniswap  where they are exposed to impermanent loss. 
We analytically characterize  the option-like payoff structures  of impermanent losses for concentrated liquidity provision and propose two static replication formulas  for    the  impermanent loss   by a combination of European calls or puts with strike prices supported on the liquidity provision price interval.   Liquidity providers could  hedge their permanent loss by trading options in more liquid centralised   exchanges such as Deribit.    The Heston stochastic  diffusion model   illustrates the extreme   accuracy of replication formulas   when there are enough traded option strikes.  Further evidences from Deribit bitcoin option market confirm the usefulness and accuracy of the static replication formulae.

\newpage 
\noindent {\bf References} \\
\bibliographystyle{apalike}
\bibliography{difi}

\end{document}